\newcommand{\longonly}[1]{#1}

\newcommand{\shortonly}[1]{}

\shortonly{\documentclass[sigconf]{acmart}}
\longonly{\documentclass[abstract=on]{scrartcl}}

\longonly{\usepackage{amsmath,amsfonts,amsthm,amssymb}\usepackage{authblk}}%

\usepackage[linesnumbered,vlined,boxed,fillcomment]{algorithm2e}
\SetKw{In}{ in }
\SetKw{And}{ and }
\SetKw{All}{ all }
\SetKw{Or}{ or }
\SetKw{Break}{break}
\SetKw{Continue}{continue}
\SetKw{KwFrom}{from }
\makeatletter
\let\oldnl\nl% Store \nl in \oldnl
\newcommand{\nonl}{\renewcommand{\nl}{\let\nl\oldnl}}% Remove line number for one line
\makeatother

\usepackage{enumerate}%
\usepackage{tikz}%
\usetikzlibrary{calc,positioning}
\usetikzlibrary{decorations.markings}%

\usepackage{xcolor}

\usepackage{hyperref}%
\hypersetup{% gestion des liens
  colorlinks=true,% colorise les liens
  breaklinks=true,% permet le retour a la ligne dans les liens trop longs
  urlcolor= blue,% couleur des hyperliens
  linkcolor= blue,% couleur des liens internes
  bookmarksopen=true,% si les signets Acrobat sont crees, les afficher completement.
  pdftitle={Slope Factorisation},% informations apparaissant dans
  pdfauthor={Adrien Poteaux},% dans les informations du document
}

\longonly{\usepackage[numbers,sort&compress]{natbib}}% pour citer [1-5] a la place de [1,2,3,4,5] par exemple, mais aussi avoir les citations "dans l'ordre"

\newcommand{\hbb}[1]{\ensuremath{\mathbb{#1}}}% anneaux / corps
\newcommand{\Ai}{\hbb{A}}

\newcommand{\Fi}{\hbb{F}}
\newcommand{\Ki}{\hbb{K}}
\newcommand{\Li}{\hbb{L}}
\newcommand{\Ni}{\hbb{N}}
\newcommand{\Qi}{\hbb{Q}}

\newcommand{\Zi}{\hbb{Z}}

\newcommand{\hcal}[1]{\ensuremath{\mathcal{#1}}}

\newcommand{\Nc}{\hcal{N}}
\newcommand{\Oc}{\hcal{O}}
\newcommand{\Rc}{\hcal{R}}
\newcommand{\Pc}{\hcal{P}}

\newcommand{\hfrak}[1]{\ensuremath{\mathfrak{#1}}}
\newcommand{\h}{\hfrak{h}}%
\newcommand{\uk}{\hfrak{u}}%
\newcommand{\vf}{\delta}

\newcommand{\hbf}[1]{\ensuremath{\mathbf{#1}}}

\newcommand{\tb}{\hbf{t}}

\newtheorem{thm}{Theorem}
\newtheorem{prop}{Proposition}
\newtheorem{cor}{Corollary}
\newtheorem{lem}{Lemma}
\newtheorem{dfn}{Definition}

\newtheorem{rmk}{Remark}
\newtheorem{hyp}{Assumption}

\newcommand{\onestep}{\textnormal{\texttt{HenselStep}}}% lemme de Hensel : val > N => val > 2*N
\newcommand{\quorem}{\textnormal{\texttt{QuoRem}}}% euclidean division
\newcommand{\fastirred}{\textnormal{\texttt{FastIrreducible}}}% irreducible test en delta
\newcommand{\slopefacto}{\textnormal{\texttt{SlopeFacto}}}
\newcommand{\divtype}{\textnormal{\texttt{DividingType}}}
\newcommand{\fastfacto}{\textnormal{\texttt{Factorisation}}}

\newcommand{\assign}{{\;\leftarrow{}\;}}% 

\renewcommand{\O}{\textrm{\Oc}}% O
\newcommand{\Ot}{\O\tilde\,\,}% soft-O
\newcommand{\Oe}{\O_\eps}% soft-O
\shortonly{}% factorisation univariee
\longonly{}% factorisation univariee
\newcommand{\F}{\textup{\textsf{F}}}%
\newcommand{\dx}{{d}}% degre en Y de F
\renewcommand{\char}{\textrm{char}}% caracteristique du corps

\newcommand{\res}[1][x]{\text{Res}_{#1}}% resultant, en Y par defaut
\newcommand\disc[1][x]{{\text{Disc}_{#1}}} % je pense qu'il ne faut pas utiliser \Delta pour ne pas s'embrouiller avec les aretes du polygone de Newton
\newcommand{\val}{\ensuremath{v}}% valuations augmentees
\newcommand{\vx}{\ensuremath{\upsilon}}% valuation dans \Ai[x]
\newcommand{\wx}{\ensuremath{w}}
\newcommand{\ordre}[1][t]{\ensuremath{\mbox{ord}_{#1}}}% ordre
\newcommand{\vF}[1][]{{\vf{}_{#1}}}% valuation du discriminant de F et F_Y
\newcommand{\Card}{\textrm{Card}}

\newcommand{\eps}{\epsilon}% pour changer facilement \epsilon / \varepsilon
\newcommand{\lc}[2]{\mbox{\rmfamily lc}_{#1}(#2)}% leading coefficient
\newcommand{\valtronc}[3][]{\lceil #2 \rceil_{#1}^{#3}}% troncation associe a une valuation

\newcommand{\Gt}{\tilde{G}}%sortie Hensel
\newcommand{\Ht}{\tilde{H}}%sortie Hensel
\newcommand{\Ut}{\tilde{U}}
\newcommand{\Vt}{\tilde{V}}
\newcommand{\Rt}{\tilde{R}}% modified residual operator
\newcommand{\At}{\tilde{A}}%sortie Hensel
\begin{document}

\title{Local polynomial factorisation: improving the Montes algorithm}

\longonly{
  %\author{%
  % \alignauthor
  \author[1]{Adrien Poteaux}% \orcid{0000-0002-7493-3001}
  \author[2]{Martin Weimann}
  \affil[1]{Univ. Lille, CNRS, Centrale Lille, UMR 9189 CRIStAL, F-59000 Lille, France}%
  % \city{Lille}
  % \country{France}
  \affil[2]{LMNO, Universit\'e de Caen-Normandie}
%  \email[1]{adrien.poteaux@univ-lille.fr}
%  \email[2]{martin.weimann@unicaen.fr}
  \date{}
}

\shortonly{
% Adrien
  \author{Adrien Poteaux}
  \orcid{0000-0002-7493-3001}
  \affiliation{%
    \institution{Univ. Lille, CNRS, Inria, Centrale Lille, UMR 9189 CRIStAL}
    \postcode{F-59000}
    \city{Lille}
    \country{France}
  }
  \email{adrien.poteaux@univ-lille.fr}
% Martin
  \author{Martin Weimann}
  % orcid{}
  \affiliation{LMNO, Universit\'e de Caen-Normandie}
  \email{martin.weimann@unicaen.fr}
}

\shortonly{%
}

\longonly{\newenvironment{CCSXML}{}}

\begin{CCSXML}
  \shortonly{
    <ccs2012>
    <concept>
    <concept_id>10010147.10010148.10010149</concept_id>
    <concept_desc>Computing methodologies~Symbolic and algebraic algorithms</concept_desc>
    <concept_significance>500</concept_significance>
    </concept>
    </ccs2012>
  }
\end{CCSXML}
%}

\shortonly{
}

\maketitle

\begin{abstract}
  We improve significantly the Nart-Montes algorithm for factoring
  polynomials over a complete discrete valuation ring $\Ai$. Our first
  contribution is to extend the Hensel lemma in the context of
  generalised Newton polygons, from which we derive a new divide and
  conquer strategy. Also, if $\Ai$ has residual characteristic zero or
  high enough, we prove that approximate roots are convenient
  representatives of types, leading finally to an almost optimal
  complexity both for irreducibility and factorisation issues, plus
  the cost of factorisations above the residue field. For instance, to
  compute an OM-factorisation of $F\in\Ai[x]$, we improve the
  complexity results of \cite{BaNaSt13} by a factor $\vF$, the
  discriminant valuation of $F$.
  % For instance, if $F \in \Zi_p[x]$ is separable with degree $\dx$
  % and discriminant valuation $\vF$, we compute its OM-factorisation
  % with an expected $\Ot(\dx\,\vF)$ operations over $\Fi_p$, assuming
  % $p$ small, improving the bound $\Ot(\dx^2+\dx\,\vF^{2})$
  % of \cite{BaNaSt13}.
\end{abstract}

% \tableofcontents{}

\section{Introduction}
\label{sec:intro}
Let $\Ai$ be a complete discrete valuation ring with residue field
$\Fi$ and consider $F\in \Ai[x]$, monic and separable of degree
$\dx$. The aim of this paper is to improve complexity bounds for the
factorisation of $F$.  Such a polynomial factorisation is a
fundamental task of computer algebra with various applications in
number theory and algebraic geometry. As such, our complexity results
allow to fasten various computational problems, such as Okutsu frames,
integral basis or genus of plane curves (see Section
\ref{sec:applications} for further details).

Our work is based on the seminal Montes algorithm \cite{GuMoNa12}, for
which the best known complexity is given in \cite{BaNaSt13}. In
\cite{GuMoNa10}, the authors conclude their paper by:
\begin{quote}
Probably, an optimal local factorisation algorithm would consist in the
application of the Montes algorithm as a fast method to get an Okutsu
approximation to each irreducible factor, combined with an efficient ``Hensel
lift'' routine able to improve these initial approximations by doubling the
precision at each iteration. One may speculate that Newton polygons of
higher order might also be used to design a similar acceleration procedure.
\end{quote}
With S. Pauli, Guardia and Nart answered partially to this question
thanks to the \emph{single-factor lifting} algorithm \cite{GuNaPa12},
that can be viewed as a Newton-like method to lift a single factor
with a quadratic convergence. This led to the overall complexity
analysis of \cite{BaNaSt13}.

In this paper, we answer more precisely to this question, by showing
that the classical Hensel algorithm can be adapted to the context of
Newton polygon of higher order. We also provide a new divide and
conquer strategy using this adapted Hensel algorithm, enabling us to
lift all factors of $F$ at the same time, with a complexity almost
linear in the size of the output. These two elements allow us to gain
a factor $\dx$ in comparison to the complexity result of
\cite{BaNaSt13}. Moreover, following \cite{PoWe18}, we show that when
$\char(\Fi)\nmid\dx$, we can use \emph{approximate roots} as strongly
optimal representatives of a type\footnote{see Sections
  \ref{sec:montes} and \ref{sec:irred} for the definitions of these
  terms}. This induces an irreducibility test with a complexity almost
linear in $\delta$ the valuation of the discriminant of $F$ ; see
Theorem \ref{thm:fastirred}. This improvement propagates for
factorisation with a slightly greater assumption
\begin{hyp}\label{hyp:char>d}
  $\char(\Fi)=0$ or $\char(\Fi)>\dx$
\end{hyp}
leading a complexity almost linear in $\dx\,n$ for a required
precision $n\geq\vF$ ; see Theorem \ref{thm:fastfacto}.

\paragraph{Related work.} Classical implemented algorithms for
factoring polynomials over $\Qi_p$ (see
e.g. \cite{CaGo00,Pa01,FoPaRo10,Pa10}) are based on the Zassenhaus
Round Four algorithm, suffering from loss of precision in computing
characteristic polynomials. In \cite{GuNaPa12}, the authors introduced
a new technique as a combination of the Montes algorithm
\cite{GuMoNa11, GuMoNa12} which exploits the Newton polygons of higher
order (as initiated in \cite{Pa10}), and a Newton-like single factor
lifting. Further complexity improvements are obtained in
\cite{BaNaSt13}. The present work is in the same vein, with the
notable difference that we introduce a multi factor lifting, which is
used in course of the Montes algorithm whenever a non trivial
factorisation is discovered.

For rings of Laurent series $\Ki((t))$ of characteristic zero or high
enough, Newton-Puiseux like algorithms can be used. The best
complexity in this context is softly linear \cite{PoWe17}, as in the
present paper, but more difficult to implement and slower for
irreducibility issues. This led us to introduce in \cite{PoWe18}
approximate roots \emph{\`a la Abhyankhar} \cite{Ab90,Pp02} in order
to derive a faster and easy-to-implement irreducibility test, quite
close to the algorithm \cite{BaNaSt13} \emph{\`a la Montes}, although
not dealing with the small characteristic. This was a first step
towards the present work, where we use now approximate roots in the
factorisation context, as allowed by a systematic use of our
generalised Hensel lifting. Note that the divide conquer in the
present paper is quite different than the one of \cite[Section
4.4]{PoWe17}: in particular, the initialisation of the Hensel
algorithm does not use the not yet implemented generalisation of the
half gcd algorithm described in \cite{MoSc16}.

Finally, let us insist that following \cite{GuNaPa12}, our algorithm
computes as a byproduct an Okutsu frame of each irreducible factors of
$F$, containing the most significant arithmetic informations
\cite{GuMoNa10} and closely related to various computational problems
of number theory and algebraic geometry, such as the computation of
integral basis (see Section \ref{sec:applications} for further
details).

\paragraph{Organisation of the paper.} We start by a summary of
important definitions related to the Montes algorithm in Section
\ref{sec:montes}. Then, we focus on the irreducibility test when
$\char(\Fi)\nmid\dx$ in Section
\ref{sec:irred}, leading to Theorem \ref{thm:fastirred}. In Section
\ref{sec:hensel}, we show how to adapt the Hensel algorithm in the
context of Newton polygon of higher orders. Section
\ref{sec:divideandconquer} uses this latter algorithm on a well chosen
type in order to derive a divide and conquer algorithm, leading to
Theorem \ref{thm:fastfacto}. Finally, we discuss some direct
applications in Section \ref{sec:applications}.

\paragraph{Complexity model.} Polynomials in $\Ai[x]$ considered in
this paper are supposed given in a dense representation, with
coefficients available up to an arbitrary precision (e.g. represented
as tables, as we always use truncation bounds). We use the algebraic
RAM model of Kaltofen \cite[Section 2]{Ka88}, counting only the number
of arithmetic operations in the residue field $\Fi$. We classically
denote $\O()$ and $\Ot()$ to respectively hide constant and
logarithmic factors in our complexity results ; see e.g. \cite[Chapter
25, Section 7]{GaGe13}. We additionally let $\Oe(d)=\O(d^{1+\eps(d)})$
with $\eps(\dx)\to 0$. We have $\Ot(\dx)\subset \Oe(\dx)$, and freely
speak of \emph{almost linear} in $\dx$ for both notations. Fast
multiplication in $\Fi[y]$ is used, i.e. we multiply two polynomials
of degree at most $\dx$ within $\Ot(\dx)$ operations in $\Fi$
\cite[Section 8.3]{GaGe13}. We assume that univariate factorisation
over $\Fi$ is available.
% Except this task, our algorithms are deterministic.
Intermediate finite extensions $\Fi_k$ of degree $f_{k-1}$ of $\Fi$
will occur (see Section \ref{sec:montes}), naturally represented as a
quotient of $\Fi[y_0,\ldots,y_{k-1}]$ by a triangular prime ideal
$(P_0(y_0),\ldots,P_{k-1}(y_0,\ldots,y_{k-1}))$.

\begin{lem}\label{lem:costFk}
  An operation in $\Fi_k$ takes $\Oe(f_{k-1})$ operations over $\Fi$.
\end{lem}

\begin{proof}
  If $\Card(\Fi)\ge \binom d 2$, use \cite[Theorem 4]{HoLe19}. If
  $\Card(\Fi)<\binom d 2$, proceed as in \cite[proof of Theorem
  1]{PoWe18} (roughly speaking, keep the first $i$ levels of the
  triangular set so that $\Card(\Fi)\,f_{i-1}\geq \binom d 2$ with
  $i$ minimal, and apply \cite[Theorem 4]{HoLe19} over $\Fi_i$ ; as
  $i\in\O(\log\log\dx)$, an operation in $\Fi_i$ is $\Ot(f_{i-1})$
  via \cite[Proposition 2]{Le15}).
\end{proof}

\begin{rmk}
  Since some subroutines use the triangular representation of $\Fi_k$
  (Remark \ref{rem:FieldTower}), introducing a randomised Las Vegas
  subroutine to fasten the arithmetic in $\Fi_k$ via a primitive
  representation would not a priori be sufficient to express our
  complexity results in $\Ot()$ instead of $\Oe()$, in contrast to
  \cite{PoWe17}.
\end{rmk}

%%% Local Variables:
%%% mode: latex
%%% TeX-master: "issac"
%%% End:

\section{Types and factorisation}
\label{sec:montes}
Let $\Li$ be a complete discrete valuation field,
$\val:\Li \rightarrow \Zi\cup\{+\infty\}$ any normalised and
surjective valuation on $\Li$ and $\pi$ an uniformiser. We denote by
$\Ai\subset \Li$ the ring of integers of $(\Li,\val)$ and by
$\Fi=\Ai/(\pi)$ the residue field of $\vx$. The two fields we have in
mind in this paper are $\Li=\Qi_p$ the field of $p$-adic numbers
% (with $\val=\val_p$, $\pi=p$, $\Ai=\Zi_p$ and $\Fi=\Fi_p$)
and $\Li=\Ki((t))$ the field of Laurent series over any field $\Ki$%
% (with $\val=\val_t$, $\pi=t$, $\Ai=\Ki[[t]]$ and %$\Fi=\Ki$)
.

\subsection{Types}
We start with types of order $0$, denoting the residue field
$\Fi_0:=\Fi$.
\begin{dfn}
  A type of order $0$ is $\tb_0=[P_0]$, where $P_0\in \Fi_0[y]$ is a
  monic irreducible polynomial.
\end{dfn}
For any $G\in\Li[x]$, a type of order $0$ comes together with the
Gauss valuation $\val_0(\sum_i a_i\,x^i):= \min_i(\vx(a_i))$ and the
residual polynomial operator $R_0(G):= G(y)/\pi^{\val_0(G)}\mod \pi$.

Types
$\tb_k=[P_0,(\phi_1,\lambda_1,P_1),\ldots,(\phi_k,\lambda_k,P_k)]$ of
order $k\ge 1$ are defined inductively below. If $1\le i \le k-1$, we
denote
$\tb_i=[P_0,(\phi_1,\lambda_1,P_1),\ldots,(\phi_i,\lambda_i,P_i)]$.
For any field $K$ and $P, Q\in K[y]$, we write $P(y) \sim{} Q(y)$ if
there exists $c\in K^\times$ such that $P(y)=c\,Q(y)$.  Also, we
denote $\Pc$ the semigroup of polygons (i.e. the set of all open
convex polygons of the plane, attached to finite formal sums of
sides) and $\Pc^-\subset \Pc$ the semi-group of polygons with negative
slopes (principal polygons). See \cite[Section 1.1]{GuMoNa12} for details.

Assume that types
of order $k-1$ have been defined and that we can attach to a type of
order $k-1$ a valuation $\val_{k-1}:\Li[x]\to \Zi$, a field extension
$\Fi_{k-1}$ of $\Fi$ and a residual polynomial operator
$R_{k-1}:\Li[x]\to \Fi_{k-1}[y]$. 
\begin{dfn}
  Let $k\geq 1$.
  $\tb_k=[P_0,(\phi_1,\lambda_1,P_1),\ldots,(\phi_k,\lambda_k,P_k)]$
  is a type of order $k$ if $\tb_{k-1}$ is a type of order $k-1$ and
  \begin{itemize}
  \item $\phi_k\in \Ai[x]$ is monic, irreducible and satisfies
    $R_{k-1}(\phi_k)\sim P_{k-1}$.
  \item $\lambda_k=-m_k/q_k\in \Qi^-$, with $(q_k,m_k)\in\Ni^2$
    coprime. We denote $(\alpha_k,\beta_k)$
    s.t. $\alpha_k\,q_k-\beta_k\,m_k=1$ with $0\le \beta_k < q_k$.
  \item $P_k\ne y \in \Fi_k[y]$ is monic, irreducible over
    $\Fi_k:=\Fi_{k-1}[y]/(P_{k-1})$. We let $\ell_k:=\deg(P_k)$ and
    $z_k:=y\mod P_k(y)\in\Fi_{k+1}$.
  \end{itemize}
  We will denote $e_k:=q_1\dots q_k$ and
  $f_k:=\ell_0\dots \ell_k=[\Fi_{k+1}:\Fi]$\footnote{A reader used to
    the work of Nart et al should pay attention that the notations
    $e_k$ and $f_k$ in \cite{GuMoNa12} are here denoted $q_k$ and
    $\ell_k$. We rather use $e_k$ and $f_k$ for the ramification index
    and residual degree discovered so far, following \cite{PoWe17}.}.
\end{dfn}

\subsection{Associated operators and representatives}
We fix a type
$\tb=[P_0,(\phi_1,\lambda_1,P_1),\ldots,(\phi_k,\lambda_k,P_k)]$ of
order $k\ge 1$ and detail several operators associated to it. If
$G\in \Li[x]$, we denote $G=\sum a_i' \phi_{k-1}^i$ and
$G=\sum a_i \phi_k^i$ their $\phi_{k-1}$ and $\phi_k$-adic
expansion\footnote{If $k=1$, we let $\phi_0=x$, $q_0=1$ and $m_0=0$,
  so that $\val_1=\val_0$.}.
\paragraph{Augmented valuation.}
$\val_k:\Li[x]\to \Zi$ is defined from $\val_{k-1}$ as
\begin{equation}\label{eq:defvk}
  \val_k(G):=\min_{i} (q_{k-1}\val_{k-1}(a_i' \phi_{k-1}^i)+m_{k-1} i).
\end{equation}
This is indeed an ``augmented valuation'' as introduced by Mac Lane
\cite{Ma36a,Ma36b} (see e.g. \cite[page 379]{GuMoNa12} for a detailed
explanation). Notice in particular that
$\val_k(G)=\min_{i} \val_k(a_i' \phi_{k-1}^i)$.

\paragraph{Newton polygon of higher order.} The \emph{polygon
  operator} $\Nc_k :\Li[x]\to \Pc$ associates to $G$ the lower convex
hull of $\{(i,\val_k(a_i\phi_k^i),\,\, a_i\ne 0\}$. We let
$\Nc_k^-(G)\in \Pc^-$ stands for the principal part of $\Nc_k(G)$.

\paragraph{Residual polynomial operator.} We need several intermediate
operators. We let
$S_k(G):=\{(i,j)\in \Nc_k(G),\, m_k i + q_k j \,\,
\textrm{is\,\,minimal}\}$,
$I_k(G):=\{i\in\Ni,\, \val_k(a_i\phi_k^i))\in S_k(G)\}$ and
$i_k(G):=\min(I_k(G))$ (letting $i_0(G)=0$ additionally). We let
$\tau_{k,i}:= \frac{i_{k-1}(a_i)+\beta_{k-1}
  \val_k(a_i\phi_k^i)}{q_{k-1}}\in\Zi$ (see \cite{GuMoNa12}, after
Definition 2.19). Then $R_k:\Li[x]\to \Fi_k[y]$ is defined
inductively as
\[
  R_k(G)=\sum_{i\in I_k(G)}
  z_{k-1}^{\tau_{k,i}}R_{k-1}(a_i\phi_k^i)(z_{k-1})\,
  y^{\frac{i-i_k(G)}{q_k}}.
\]

\begin{rmk}
  Let us summarise the dependencies of these operators: $\val_k$
  depends only of $\val_{k-1}, \phi_{k-1}$ and $\lambda_{k-1}$, thus
  only of $\tb_{k-1}$. $\Nc_k$ depends of $\val_k$ and
  $\phi_k$. Finally, $R_k$ depends of $R_{k-1}$, $\phi_k$ and
  $\lambda_k$.
\end{rmk}

\paragraph{Representative of $\tb_k$.} They are defined as follows.
\begin{dfn}\label{dfn:type} Let $G\in \Ai[x]$ be monic. We say that
  $G$ is of type $\tb$ if for $0\le i\le k$, $\Nc_i(G)$ is one-sided
  of slope $\lambda_i$, and for $1\le i\le k$ $R_i(G)\sim P_i^{N_i}$
  for some $N_i\in \Ni^\times$. We denote by $G_\tb$ the product of
  all monic irreducible factors of $G$ of type $\tb$. We say that
  $\tb$ divides $G$ is $\deg(G_\tb)>1$. Finally, $G$ is said to be a
  representative of $\tb$ if additionally
  $\ordre[\tb]{(G)}:=\ordre[P_k]{R_k(G)}=1$.
\end{dfn}

\subsection{Factorisation according to a type}
The following theorem summarises the main results that led to the Montes
algorithm:
\begin{thm}\label{thm:polygon-residue} Let $k\ge 1$ and $\tb$ be a
  type of order $k-1$ together with a representative $\phi_k$. Denote
  $\val_k$ the associated augmented valuation and assume $F\in \Ai[x]$
  monic.
  \begin{enumerate}
  \item (Theorem of the polygon, \cite[Theorem 3.1]{GuMoNa12}) Suppose
    that $\Nc_k^-(F)=S_1+\cdots +S_g$ where the polygons
    $S_1,\ldots,S_g$ are one-sided of distinct slopes
    $\lambda_{k,1},\ldots,\lambda_{k,g}$. Denote by $R_{k,i}$ the
    residual polynomial operator associated to $\val_k$, $\phi_k$ and the
    slope $\lambda_{k,i}$. The polynomial $F_\tb$ admits a
    factorisation
    \[
      F_\tb=F_{\tb,1}\cdots F_{\tb,g} \in \Ai[x]
    \]
    where $\Nc_k^-(F_{\tb,i})=S_i$ up to translation and
    $R_{k,i}(F_{\tb,i})\sim R_{k,i}(F)$.
  \item (Theorem of the residual polynomial, \cite[Theorem 3.7]{GuMoNa12})
    Let $R_{k,i}(F)\sim P_{k,i,1}^{a_1} \cdots P_{k,i,r}^{a_r}$ where
    the $P_{k,i,j}$ are pairwise coprime irreducible polynomials. Then $F_{\tb,i}$ admits a
    factorisation
    \[
      F_{\tb,i}=F_{\tb,i,1}\cdots F_{\tb,i,r_i} \in \Ai[x]
    \]
    where $\Nc_k(F_{\tb,i,j})$ is straight of slope $\lambda_{k,i}$
    and $R_{k,i}(F_{\tb,i,j})\sim P_{k,i,j}^{a_j}$.
  \item (Irreducibility criterion) If $a_j=1$, then $F_{\tb,i,j}$ is
    irreducible.
  \end{enumerate}
\end{thm}

%%% Local Variables:
%%% mode: latex
%%% TeX-master: "issac"
%%% End:

\section{Testing irreducibility}
\label{sec:irred}
\subsection{Approximate roots}
\label{ssec:approx}

\begin{prop}\label{prop:approx-root} Let $F\in\Ai[x]$ be monic of
  degree $\dx$, with $\char(\Ai)\nmid\dx$. Let $N\in \Ni$ dividing
  $\dx$. There exists a unique polynomial $\psi\in \Ai[x]$ monic of
  degree $d/N$ such that $\deg(F-\psi^{N}) < d-d/N$. We call it the
  $N^{th}$ approximate roots of $F$, denoted by $\sqrt[N]{F}$. It can
  be computed in less than $\Ot(\dx)$ operations in $\Ai$.
\end{prop}

\begin{proof}
  See e.g. \cite[Proposition 3.1]{Pp02} for the existence, and
  \cite[Proposition 11]{PoWe18} for their computation.
\end{proof}

\begin{prop}\label{prop:key} Let $F\in \Ai[x]$ be a monic separable
  polynomial of type $\tb$ and denote $N=\ordre[\tb](F)$. Suppose
  $\char(\Fi)\nmid N$ and let $\psi=\sqrt[N]{F}$.
  \begin{enumerate}
  \item \label{enum:key1} $\psi$ is a representative of $\tb$.
  \item \label{enum:key2} If $F$ is of type $\tb\cup (\psi,-m/q,P)$,
    then $q\deg(P)>1$.
  \end{enumerate}
\end{prop}

\begin{proof}
  Point \ref{enum:key1} can be shown with arguments similar to
  \cite[Lemma 7]{PoWe18}. Point \ref{enum:key2} is a direct
  consequence of the fact that the coefficient of $\psi^{N-1}$ in
  the $\psi$-adic expansion of $F$ is zero.
\end{proof}

Without dealing with the precision of computations in $\Ai$, this
leads to the following irreducibility test algorithm:

\begin{algorithm}[ht]
  \nonl\TitleOfAlgo{\fastirred($F$)\label{algo:Irreducible}}%
  \KwIn{ $F\in \Ai[x]$ monic separable s.t.
    $\char(\Fi)\nmid\deg(F)$.}%
  \KwOut{A Boolean (is $F$ irreducible ?), a type $\tb$ and a
    representative $\phi$ of $\tb$.}%
  \lIf{$R_0(F)$ is not some $P_0^{N_0}$}{\Return{False, $[\,]$}}%
  $\tb\gets [P_0]$, $k\gets 1$, $N\gets N_0$\;%
  \While{$N > 1$}{%
    $\phi_{k}\gets \sqrt[N]{F}$\;%
    \label{algoIrr:represent}
    \lIf{$\Nc_k^-(F)$ is not one sided}{\Return{(False, $\tb$, $\phi_k$)}}%
    \lIf{$R_k(F)$ is not some $P_k^{N_k}$}{\Return{(False, $\tb$, $\phi_k$)}}%
    $\tb\gets \tb\cup (\phi_k,\lambda_k,P_k)$\tcp*{$\lambda_k$ the slope of $\Nc_k(F)$}%
    $N\gets N_k$, $k\gets k+1$\;%
  }%
  \Return{(True, $\tb, F$)}%
\end{algorithm}

This algorithm is similar to the one of Montes et al, with the
exception of the way we construct the representatives: approximate
roots enable a quick computation of the representative $\phi_k$,
together with the additional property $q_k\ell_k>1$ which ensures
$k\in\log(d)$.

\subsection{Precision and complexity}\label{ssec:prec}
It remains to deal with the necessary precision to conduct operations
in $\Ai$ and get complexity bounds for the computation of $\Nc_k^-(F)$
and $R_k(F)$.  We proceed as in \cite[Section 5.4]{PoWe18}: starting
with a small precision $\sigma$, we check at each iteration if
$\sigma$ is sufficient to certify that the computed data of
$F\mod \pi^\sigma$ is truly the data of $F$. If the precision is not
sufficient, we double it and restart the whole computation. This
process multiplies the overall complexity by at most $2$. We need a
certificate that the current precision $\sigma$ is high enough and an
upper bound for $\sigma$.

Assume that we computed a type $\tb_{k-1}$
dividing $F$ using a precision $\sigma$. Compute $\Nc_k^{-}(F)$ with
precision $\sigma$ and denote $\lambda_{\min}$ its \emph{right hand
  slope}, with convention $\lambda_{\min}(F)=+\infty$ if it is reduced
to a vertex.

% \begin{figure}
%   \begin{tikzpicture}
%     \small
%     % axes
%     \draw (-.5,0) node[left] {$\val_k(F)$} -- (5.5,0);%
%     \draw (0,-.5) node[below] {$0$} -- (0,3.2);%
%     % fleches
%     \flecheN{(0,3.2)} \flecheE{(5.5,0)}%
%     % polygone de Newton (avant les points pour recouvrement)
%     \draw[thick] (3,1) -- (5,0) node[near start, above right]
%     {$\lambda_{\min}$} node [below] {$N$};%
%     % points du support
%     \foreach\i in {(3,1),(5,0)}{\ptA[black]{\i}}%
%     % numeros et lignes pointillees
%     \draw[dashed] (-.2,2) -- (5.3,2) node[at start, left]
%     {$\sigma\,\val_k(\pi)$};% borne de troncation
%     % prolongement des aretes
%     \draw[dotted] (3,1) -- (0,2.5) node[left] {$\val_k(F)-\lambda_{\min}\,N$};%
%   \end{tikzpicture}
%   \caption{\label{pic:truncation}Necessary precision to compute $\Nc_k(F)$}
% \end{figure}

\begin{lem}\label{lem:trunc}
  Let $F\in \Ai[x]$ monic divisible by $\tb_{k-1}$.  If
  $\sigma>\frac{\val_k(F)+N\,|\lambda_{\min}|}{\val_k(\pi)}$, then
  truncating computations modulo $\pi^{\sigma+1}$ will compute the
  correct right hand edge of $\Nc_k^-(F)$. If moreover $F$ is of type
  $\tb_{k-1}$, then
  $ \frac{\val_k(F)+N\,|\lambda_{\min}|}{\val_k(\pi)}<
  \frac{2\vF}{\dx}$.
\end{lem}

\begin{proof}
  This is \cite[Lemmas 2.9 and 2.8]{BaNaSt13}.
\end{proof}

\begin{lem}\label{lem:vk}
  One can compute $\val_k(F)$ in $\Ot(\dx)$ operations in $\Ai$.
\end{lem}

\begin{proof}
  Compute the $\phi_{k-1}$-adic expansion of $F$ in $\Ot(\dx)$
  operations in $\Ai$ \cite[Theorem 9.15]{GaGe13}. If $k>1$, compute
  recursively each $\val_{k-1} (a_i \phi_{k-1}^i)$. As there is a
  closed formula for $\val_{k-1}(\phi_{k-1})$ \cite[Proposition
  2.15]{GuMoNa12}, the bound $\deg(a_i)<\deg(\phi_{k-1})$ concludes.
\end{proof}

\begin{prop}\label{prop:Nk}
  One can compute $\Nc_k^-(F)$ with precision $\sigma$ in less than
  $\Ot(\dx\,\sigma)$ operations in $\Fi$.
\end{prop}

\begin{proof}
  First compute the $\phi_k$-adic expansion of $F$, then the different
  values $\val_k(a_i\,\phi_k^i)$. The complexity then comes from
  \cite[Theorem 9.15]{GaGe13} and Lemma \ref{lem:vk}.
\end{proof}

\begin{prop}\label{prop:Rk}
  Up to the cost of operations already done while computing $\Nc_k^-(F)$,
  one can compute $R_k(F)$ in $\Oe(\dx\,f_{k-1})$ op. in $\Fi$.
\end{prop}

\begin{proof}
  This is \cite[Lemma 5.6]{BaNaSt13}.
\end{proof}

\begin{lem}\label{lem:fd<2lvF}
  Let $F\in \Ai[x]$ monic of type $\tb_{k-1}$ such that
  $\ordre[\tb_0]{F}>1$. Then $f_{k-1}\dx \le 2\ell_0 \vF$.
\end{lem}

\begin{proof}
  We know that $F\equiv P_0^{N_0}\mod \pi$ for some irreducible
  polynomial $P_0\in \Fi[x]$ of degree $\ell_0$. Let $\theta$ be a
  root of $F$, say with minimal polynomial $G$ (prime factor of
  $F$). Thus $\theta \mod \pi$ is a root of $P_0$ and there are
  $N_0-1$ other roots $\theta'$ of $F$ such that
  $\theta \equiv \theta'\mod \pi$. It follows that
  $\val_0(\theta-\theta')>0$, hence
  $\val_0(\theta-\theta')\ge \frac 1 {e_{G}}$ where $e_G$ is the
  ramification index of $G$. Summing over all $\theta$ and all
  $\theta'$, we get $\vF\ge \sum_{G} \deg(G)(N_0-1)/e_{G}$, the sum
  over all prime factors of $F$. We have $e_{G}f_{G} =\deg(G)$ ($f_G$
  residual degree) and $f_{k-1}|f_G$ for all $G$. It follows that
  $\vF\ge f_{k-1}(N_0-1)$. Equality $N_0\,\ell_0=\dx$ concludes.
\end{proof}

\begin{thm}\label{thm:fastirred}
  If $\char(\Fi)\nmid\dx$, then there exists an algorithm that tests
  if $F$ is irreducible in less than $\Oe(\vF\ell_0)$ operations in
  $\Fi$ and at most $\log_2(d)$ univariate irreducibility tests.
\end{thm}

\begin{proof}
  The algorithm is \fastirred{} with the above modification to deal
  with the precision of the computations. As we use a precision
  $\sigma\leq 2\vF/\dx$, the complexity comes from all the
  intermediate results of this section.
\end{proof}

\paragraph{The case $\Ai=\Zi_p$.} We say that $F\in\Zi_p[x]$ is
Weierstrass if $F\equiv x^{\dx}\mod p$. Also, we say that $p$ is small
if it is polynomial in $\dx$ and $\vF$.
%, i.e $p\in \O((\dx+\vF)^{\Ot(1)})$.

\begin{cor}
  If $p$ is small and does not divide $\dx$, we can test the
  irreducibility of a separable Weierstrass polynomial $F\in\Zi_p[x]$
  with $\Oe(\vF)$ operations in $\Fi_p$.
\end{cor}

\begin{proof}
  $F$ being Weierstrass, we have $\ell_0=1$ and $\dx\le 2\vF$ by Lemma
  \ref{lem:fd<2lvF}. We can check if $R_k(F)\in \Fi_k[x]$ is a prime
  power (and then compute $P_k$) within $\Oe(\dx\log(p))$ operations
  in $\Fi_p$ using \cite[Corollary 2]{HoLe20}\footnote{This result
    extends to towers of fields following the proof of \cite[Corollary
    3]{HoLe20}} together with $\deg(R_k(F))f_{k-1}\leq\dx$.
\end{proof}
Under the same hypothesis, Montes irreducibility test would require
$\Oe(\vF^2)$ operations in $\Fi_p$, see \cite[Theorem 5.10]{BaNaSt13}.

\subsection{The case $\char(\Fi)\mid \dx$}
\label{ssec:montes-irr}

When $\char(\Fi)\mid\dx$, approximate roots cannot be used as
representatives of types. Following \cite{BaNaSt13}, we compute these
representatives in another way (Proposition \ref{prop:construct}
below, in $\Oe(\dx\wx(F)/\wx(\pi))\subset\Oe(\vF)$ operations in
$\Fi$). Unfortunately, we might have now $q_k\ell_k=1$ (refinement
steps) and the number of iterations is not in $\O(\log(\dx))$
anymore. From \cite[Lemma 5.11]{BaNaSt13}, we can bound the number of
recursive call by $\O(\vF/\ell_0)$\footnote{Equation (5.4) of
  \cite{BaNaSt13} actually shows that it is bounded by
  $\O(\vF/f_{k-1})$}. This leads to a complexity in $\Oe(\vF^2)$
operations in $\Fi$, plus the univariate irreducibility tests. We can
still run only $k\leq\log_2(\dx)$ of them by first checking that
$R_k(F)$ has a single root (i.e. $q_k\ell_k=1$), using a univariate
shift, in $\Ot(\dx)$ operations in $\Fi$, and use the univariate
irreducibility test only when $q_k\ell_k>1$. We proved the following:
\begin{prop}\label{prop:GenIrr}
  If $\char(\Fi)$ divides $\dx$, one can check if $F$ is irreducible
  in less than $\Oe(\vF^2)$ operations in $\Fi$ and at most
  $\log_2(\dx)$ univariate irreducibility test.
\end{prop}
This result is very similar to \cite[Theorem 5.10]{BaNaSt13} (we just
use some minor complexity improvements in some intermediate results,
described in Section \ref{ssec:prec}). We still call \fastirred{} the
underlying algorithm.

%%% Local Variables:
%%% mode: latex
%%% TeX-master: "issac"
%%% End:

\section{Generalised slope factorisation}
\label{sec:hensel}
Algorithm \fastirred{} of the previous section will either prove that
$F$ is irreducible, either provide a type $\tb_{k-1}$ together with a
representative $\phi_k$ such that either $\Nc_k(F)$ has two or more
distinct slopes, either $R_k$ is not a power of an irreducible
polynomial. To summarise these two possibilities, it is convenient to
consider a slight modification of the residual polynomial operator
$R_k$ attached to the \textit{right hand slope} $\lambda_k=-m_k/q_k$
of $\Nc_k^{-}(F)$:
\begin{dfn}\label{dfn:Rtilde}
  The \emph{modified residual polynomial} of $G\in \Li[x]$ (attached
  to $\lambda_k$) is defined by
  $\Rt_k(G)(y):=y^{i_k(G)} R_k(G)(y^{q_k})$.
\end{dfn}
If $F$ is not irreducible, we get
\[
  \Rt_k(F) = h_0 h_1\cdots h_r
\]
with $h_0\sim{} y^{i_k(F)}$, $h_1,\dots,h_r\in\Fi_k[y^{q_k}]$ coprime,
this factorisation containing at least two different non trivial
factors.

From Theorem \ref{thm:polygon-residue}, there exist
$F_0,F_1,\cdots,F_r\in\Ai[x]$ monic such that
\[
  F=F_0 F_1 \cdots F_r \quad \text{ with }\quad \Rt_k(F_i)\sim{}h_i,\quad
  0\leq i\leq r.
\]
This section describes a Hensel-like algorithm to compute such a
factorisation up to an arbitrary precision with complexity almost
linear in the size of the output.

\subsection{A (slightly) more general problem}
\label{ssec:pb-hensel}
We express our problem in a slightly more general context, that will
be useful in Section \ref{sec:divideandconquer}. Let $F\in \Ai[x]$ be
a monic polynomial, $\tb=\tb_{k-1}$ be a type of order $k-1$ dividing
$F$ and $\phi=\phi_k$ a representative of $\tb$. We assume that either
$F_{\tb}$ is a proper factor of $F$, either $F_{\tb}$ admits a non
trivial factorisation at order $k$ induced by Theorem
\ref{thm:polygon-residue}.

Let $\val=\val_k$ the augmented valuation built from $\tb$ (i.e. from
$\val_{k-1}, \phi_{k-1}$ and $\lambda_{k-1}$) and $\Nc(F)=\Nc_k(F)$
the generalised Newton polygon defined by $\val$ and $\phi$. Denote
$\lambda=-\frac{m}{q}\in \Qi^-$ the \textit{right hand slope} of
$\Nc^-(F)$ and $R(F)$ (resp. $\Rt(F)$) the (modified) residual
polynomial defined by $\tb$, $\phi$ and $\lambda$.

\begin{lem}\label{lem:leadcoeff}
  Let $G\in \Ai[x]$ monic and $g=\Rt(G)$. If $G$ is of type $\tb$,
  then
  \[
    \Nc(G)=\Nc^-(G),\,\,
    \deg(G)=\deg(\phi)\deg(g),\,\,\lc{}{g}=R(\phi)^{\deg(g)}.
  \]
  If $\tb$ does not divide $G$ then $g\in \Fi_k^\times$.
\end{lem}

\begin{proof}
  See e.g. \cite{GuMoNa12}.
\end{proof}

Thanks to this lemma and our assumption on $F$, we get
\[
  \Rt(F) = h_0 h_1\cdots h_r h_\infty
\]
with $h_0=R(\phi)^s y^s$, $h_\infty \in \Fi_k^\times$ and
$h_1,\ldots,h_r\in \,\Fi_k[y^q]$ powers of some coprime irreducible
polynomials satisfying $h_i(0)\neq 0$ and
$\lc{}{h_i}=R(\phi)^{\deg(h_i)}$. Moreover, there are uniquely
determined monic polynomials $F_0,\ldots, F_r,F_{\infty}\in \Ai[x]$
such that
\[
  F=F_0 F_1 \cdots F_r F_\infty \quad \text{ with }\quad
  \Rt(F_i)=h_i,\quad i=0,\ldots,t,\infty.
\]
In particular, we have $F_\tb=F_0 F_1 \cdots F_r$ and $F=F_\tb F_\infty$.

Such a factorisation can be seen as a generalisation of the ``slope
factorisation'' of \cite{CaRoVa16}. It is natural to express the
precision using the augmented valuation $\wx=\val_{k+1}$ defined by
$\val$, $\phi$ and $\lambda$ following \eqref{eq:defvk}. Thanks to a
dichotomic argument, we are reduced to compute $F=\Gt \Ht$ (up to some
precision w.r.t. $\wx$), with $\Gt, \Ht$ products of some of the
$F_i$'s above. We proceed as follows:
\begin{enumerate}
\item \textbf{Initialisation.} Compute $G, H\in \Ai[x]$ and
  $U, V\in\Li[x]$ such that $\wx(F-G\,H)>\wx(F)$ and
  $\wx(U\,G+V\,H-1)>0$.
\item \textbf{Lifting.} Run the classical Hensel Lemma with adapted
  truncations to compute $\Gt$ and $\Ht$ such that
  $\wx(F-\Gt\Ht)\ge \wx(F)+n$ for a given precision $n$.
\end{enumerate}

\subsection{Initialisation}\label{ssec:slopeinit}

A key point is to compute polynomials with prescribed (modified)
residual polynomial. We start with a definition.
  
\begin{dfn}\label{dfn:monic}
  A polynomial $H\in\Ai[x]$ is said \emph{monic in $\phi$} if it has
  $\phi$-adic expansion $H=\sum_{i=0}^N a_i\phi^i$ with $a_N=1$. We
  say that $H$ is strongly monic in $\phi$ (with respect to $\wx$) if
  moreover $\wx(H)=N\wx(\phi)$.
\end{dfn}

\begin{prop}\label{prop:construct}
  Let $h=y^s\,\h(y^q)$ with $\h\in\Fi_k[y]$ and $W\in\Zi$.
  We can compute $H\in \Li[x]$ of smallest degree such that
  \[
    \Rt(H)=h,\quad \wx(H)=W, \quad \deg(H)< (\deg(h)+1)\deg(\phi).
  \]
  It takes
  $\Ot\left(\deg(H)(\deg(h)+1)\,\frac{\wx(\phi)}{\wx(\pi)}\right)$
  operations in $\Fi$. Furthermore:
  \begin{enumerate}
  \item If $W\ge \deg(h) \wx(\phi)$, then $H\in\Ai[x]$.
  \item If $W=\deg(h) \wx(\phi)$ and $\lc{}{h}=R(\phi)^{\deg(h)}$,
    then $H$ is strongly monic in $\phi$.
  \end{enumerate}
\end{prop}

\begin{proof}
  First use \cite[Lemma 5.7]{BaNaSt13} to compute $G\in\Ai[x]$ such that
  $\Rt(G)=h$ and $\wx(G)=W+\wx(\pi)\,n$ with
  $n=\left\lceil \frac{\deg(h) \wx(\phi)
      -W}{\wx(\pi)}\right\rceil$. We can divide their complexity
  result by a factor $\deg(h)$ by replacing Horner evaluation therein
  by a divide and conquer strategy (see e.g. \cite{HaNo11}). Finally,
  output $H=\pi^{-n}\,G$.
\end{proof}

\begin{rmk}\label{rem:FieldTower}
  The representation of $\Fi_k$ as a tower of fields as described
  before Lemma \ref{lem:costFk} is used in the proof of \cite[Lemma
  5.7]{BaNaSt13}. In particular, no operations in $\Fi_k$ is
  performed, so that the complexity can be expressed with $\Ot()$ and
  not $\Oe()$.
\end{rmk}

\begin{lem}\label{lem:Rt-w}
  Let $G,H\in \Li[x]$ such that $\wx(G)=\wx(H)$.
  \begin{enumerate}
  \item  $\tilde{R}_k(G)=\tilde{R}_k(H)$ if and only if $\wx(G-H)>\wx(G)$.
  \item If $\wx(G+H)=\wx(G)$, then $\tilde{R}_k(G+H)=\tilde{R}_k(G)+\tilde{R}_k(H)$.
  \end{enumerate}
\end{lem}

\begin{proof}
  (1) Since $R(G)$ and $R(H)$ have non zero constant term, Definition
  \ref{dfn:Rtilde} implies that $\tilde{R}(G)=\tilde{R}(H)$ if and
  only if $R(G)=R(H)$ and $i_k(G)=i_k(H)$. As $\wx(G)=\wx(H)$, this is
  equivalent to that $R(G)=R(H)$ and $S(G)=S(H)$. \cite[Proposition
  2.8]{GuMoNa12} concludes.

  (2) Equality $\wx(G+H)=\wx(G)$ implies that
  $S_k(G),S_k(H), S_k(G+H)$ lie on a same segment $T$ of slope
  $\lambda$. We conclude with \cite[Lemma 2.23 and eq. (19)]{GuMoNa12},
  together with Definition \ref{dfn:Rtilde}.
\end{proof}

Let us consider now our initialisation problem. By the discussion
above, we can assume that $\Rt(F)=g\,h$ with $g,h$ coprime,
$g\in\Fi_k[y^q]$ and $h(y) =y^s\,\h(y^q)$, with $\h\in \Fi_k[y]$,
$\h(0)\ne 0$. We can additionally assume that
$\lc{y}{h}=R(\phi)^{\deg(h)}$. Notice that we might have
$g\in \Fi_k^\times$.

\begin{prop}\label{prop:init}
  Let $g,h$ as above. One can compute $G, H\in \Ai[x]$ and
  $U, V\in\Li[x]$ such that $H$ is strongly monic in $\phi$,
  $\Rt(G)=g$, $\Rt(H)=h$, $\deg(G)+\deg(H)\le \deg(F)$,
  $\wx(F-G\,H)>\wx(F)$, $\deg(U)<\deg(H)$, $\deg(V)<\deg(G)$ and
  $\wx(U\,G+V\,H-1)>0$ in less than $\Ot(\dx\wx(F)/\wx(\pi))$
  operations in $\Fi$.
\end{prop}

\begin{proof}
  Use Proposition \ref{prop:construct} to compute $G,H\in \Ai[x]$ with
  $H$ strongly monic in $\phi$ and such that $\tilde{R}(H)=h$,
  $\tilde{R}(G)=g$, $\wx(F)=\wx(GH)$ with $\Ot(\dx \wx(F)/\wx(\pi))$
  operations in $\Fi$ (this bound being a consequence of
  $\wx(F)=\deg(\Rt(F))\,\wx(\phi)$). Notice that
  $\deg(G)+\deg(H)\le \deg(F)$ as we construct $G$, $H$ of minimal
  degrees. As $\wx(GH)=\wx(F)$ and $\tilde{R}(GH)=\tilde{R}(F)$, point
  1 in Lemma \ref{lem:Rt-w} gives $w(F-GH)>w(F)$. For $U,V$, we first
  compute $u,v$ in $\Fi_k[y]$ such that $u\,g+v\,h=1$,
  $\deg(u)<\deg(h)$, $\deg(v)<\deg(g)$ with
  $\Oe(\deg(gh)\,f_{k-1})\subset \Oe(\dx)$ operations in $\Fi$
  \cite[Corollary 11.9]{GaGe13}. We necessarily have $v\in\Fi_k[y^q]$
  and $u=y^t\,\uk$ with $\uk\in\Fi_k[y^q]$, so we may apply again
  Proposition \ref{prop:construct} to compute $U,V\in \Li[x]$ such
  that $\Rt(U)=u$, $\Rt(V)=v$, $\deg(U)<\deg(H)$, $\deg(V)<\deg(G)$,
  $\wx(U)=-\wx(H)$ and $\wx(V)=-\wx(G)$ within the same complexity
  bound. We get $\tilde{R}(UG)+\tilde{R}(VH)=ug+vh=1\ne 0$ and Lemma
  \ref{lem:Rt-w} (point 1) implies that $\wx(UG+VH)=\wx(UG)=\wx(VH)$,
  so that $\tilde{R}(UG+VH)=\tilde{R}(UG)+\tilde{R}(VH)$ (point 2). We
  get $\tilde{R}(UG+VH)=1=\tilde{R}(1)$ and
  $\wx(UG+VH)=0=\wx(1)$. Point 1 again gives $\wx(UG+VH-1)>\wx(1)=0$.
\end{proof}

\subsection{Lifting: a valuated Hensel Lemma}
\label{ssec:hensel}

Let \quorem{} denotes the usual euclidean algorithm.

\begin{lem}
  \label{lem:quorem}
  Let $A$, $B\in\Li[x]$ with $B$ strongly monic in $\phi$. Then,
  $Q,R=\quorem(A,B)$ satisfies $\wx(R)\geq \wx(A)$ and
  $\wx(Q)\geq \wx(A)-\wx(B)$.
\end{lem}

\begin{proof}
  We focus on the computation of $R$.  First note that it be computed
  as follows\footnote{in practice, we use the classical algorithm of
    $\Ai[x]$, this is only for this proof purpose}: write
  $A=\sum_{i=0}^N a_i\,\phi^i$ the $\phi$-adic expansion of $A$ and
  $B=\phi^b+\cdots$. If $N<b$, we get $R=A$ ; otherwise, compute
  $\At=A-a_N\,\phi^{N-b}\,B$, and apply recursively this strategy to
  $\At$. As $\At=\sum_{i=0}^{N-1} \tilde a_i\,\phi^i$, this procedure
  converges towards the unique remainder $R$.  We now prove the result
  by induction on the value $N\geq b$. By linearity, we can assume
  $A=a_N\,\phi^N$. We then have $\wx(\At)\geq \wx(A)$ as
  $\wx(a_N\,\phi^{N-b}\,B) = \wx(A)$ from the assumption
  $\wx(B)=b\,\wx(\phi)$. By induction, this proves the lemma for
  $R$. Result for $Q$ is then a straightforward consequence, as
  $\wx(Q\,B)=\wx(A-R)$.
\end{proof}

This lemma will enable us to prove that the classical Hensel lemma
\cite[Algorithm 15.10]{GaGe13}, when starting with correct initial
polynomials, ``double the precision'' according to an extended
valuation $(\wx,\phi)$. The only difference with the classical
algorithm is the way we truncate polynomials: for any polynomial
$F\in\Ai[x]$ and $n\in\Zi$, we denote
$\valtronc[]{F}{n}=\valtronc[k+1]{F}{n}$ the ``truncation of $F$
according to the valuation $\wx=\val_{k+1}$'', defined recursively as
follows. We let $\valtronc[0]{F}{n}$ the usual truncation of $F$ up to
precision $\pi^n$ and, if $F=\sum_i a_i\,\phi_k^i$ and $k\ge 0$, then
$\valtronc[k+1]{F}{n}=\sum_i
\valtronc[k]{a_i}{\frac{n-\val_{k+1}(\phi_k^i)}{q_k}}\,\phi_k^i$. This
is indeed a natural definition, as $\deg(a_i)<\deg(\phi_k)$ implies
$\val_{k+1}(a_i)=q_k\,\val_k(a_i)$. In other words, we remove all
terms of $F$ that have $\wx$-valuation greater than $n$ in its
$(\phi_0,\ldots,\phi_k)$-multiadic expansion \cite[Section 3]{PoWe18}.

\begin{lem}\label{lem:complexityL[x]}
  Let $G\in \Li[x]$ with $\deg(G)\leq\dx$ and $n\in \Ni$. We can
  compute $\valtronc{G}{n}$ in $\Ot(\dx(n/\wx(\pi)-\val_0(G)))$
  operations in $\Fi$.
\end{lem}

\begin{proof}
  Compute $\valtronc{\pi^{-\val_0(G)}G}{n-\val_0(G)\wx(\pi)}$ with
  precision $\lceil n/\wx(\pi)\rceil-\val_0(G)$ following the
  recursive definition above. The main cost is to compute the
  $\phi_k$-adic expansions, as in the proof of Lemma \ref{lem:vk}.
\end{proof}

Algorithm \onestep{} below takes as input $F,G,H\in \Ai[x]$ with $H$
strongly monic in $\phi$, $U,V\in \Li[x]$ and $n\in \Ni^\times$ such
that
\begin{itemize}
\item $\wx(F-G\,H)\geq\wx(F)+n$, with $\deg(F)\geq\deg(G)+\deg(H)$
\item $\wx(U\,G+V\,H-1)\geq n$, with $\deg(U)<\deg(H)$,
  $\deg(V)<\deg(G)$, $\wx(U)=-\wx(G)$ and $\wx(V)=-\wx(H)$.
\end{itemize}
It computes $\Gt,\Ht\in \Ai[x]$ with $\Ht$ strongly monic in $\phi$
and $\Ut,\Vt\in \Li[x]$ such that
\begin{itemize}
\item $\wx(F-\Gt\,\Ht)\geq\wx(F)+2\,n$, with $\deg(\Ht)=\deg(H)$,
  $\deg(F)\geq\deg(\Gt)+\deg(\Ht)$, $\wx(\Ht-H)\geq\wx(H)+n$ and
  $\wx(\Gt-G)\geq\wx(G)+n$.
\item $\wx(\Ut\,\Gt+\Vt\,\Ht-1)\geq 2\,n$, with $\deg(\Ut)<\deg(\Ht)$,
  $\deg(\Vt)<\deg(\Gt)$, $\wx(\Ut)=-\wx(\Gt)$ and
  $\wx(\Vt)=-\wx(\Ht)$.
\end{itemize}

\begin{algorithm}[ht]
  \nonl\TitleOfAlgo{\onestep($F,G,H,U,V,n$)\label{algo:OneStep}}
  $E\assign{}\valtronc{F-G\, H}{\wx(F)+2n}$\;%
  $Q,R\assign{}\valtronc{\quorem(U\,E,H)}{\wx(F)+2n}$\;%
  $\Gt\assign{}\valtronc{G+E\,{}V+Q\, G}{\wx(G)+2n}$\;%
  $\Ht\assign{}\valtronc{H+R}{\wx(H)+2n}$\;%
  $B\assign{}\valtronc{U\,{}\Gt+V\,\Ht-1}{2n}$\;%
  $C,D\assign{}\valtronc{\quorem(U\,B,\Ht)}{2n}$\;%
  $\Ut\assign{}\valtronc{U-D}{2n-\wx(G)}$\;%
  $\Vt\assign{}\valtronc{V-B\,{}V-C\,\Gt}{2n-\wx(H)}$\;%
  \Return{$\Ht$, $\Gt$, $\Ut$, $\Vt$}%
\end{algorithm}

\begin{prop}
  \label{prop:hensel}
  Algorithm \onestep{} is correct. It takes less than
  $\Ot\left(\frac{n+\wx(F)}{\wx(\pi)}\dx\right)$ operations in $\Fi$.
\end{prop}
\begin{proof}
  We start with the correctness. $H$ being strongly monic in $H$,
  Lemma \ref{lem:quorem} ensures $\wx(R)\geq{}\wx(H)+n>\wx(H)$ as
  $\wx(E)\geq{}\wx(F)+n$ by assumption. As $\deg(R)<\deg(H)$, this
  proves $\deg(\Ht)=\deg(H)$, $\wx(\Ht-H)\geq\wx(H)+n$ and $\Ht$
  strongly monic in $\phi$. As $\wx(Q)\geq{}n$, we also get
  $\wx(\Gt-G)=\wx(VE+QG)\geq\wx(G)+n$. Using these results and the
  equality
  \[
    \valtronc{F-\Gt\,\Ht}{\wx(F)+2n} =
    \valtronc{E\,(1-U\,G-V\,H)-R\,(E\,V+Q\,G)}{\wx(F)+2n},
  \]
  we get $\valtronc{F-\Gt\,\Ht}{\wx(F)+2n} =0$,
  i.e. $\wx(F-\Gt\,\Ht)\geq\wx(F)+2\,n$, from which we also deduce
  $\deg(F)\geq\deg(\Gt)+\deg(\Ht)$.

  Similarly, using $\valtronc{C\,\Ht+D}{2n}=\valtronc{U\,B}{2n}$, we
  get
  \[
    \valtronc{\Ut\,\Gt+\Vt\,\Ht-1}{2n} =
    \valtronc{U\,\Gt+V\,\Ht-1-B\,(U\,\Gt+V\,\Ht)}{2n} =
    \valtronc{B^2}{2n}=0,
  \]
  proving $\wx(\Ut \,\Gt+\Vt\,\Ht-1)\geq 2\,n$. Using Lemma
  \ref{lem:quorem} once again, we can deduce easily the remaining
  properties of the output.

  Finally, the complexity is a direct consequence of Lemma
  \ref{lem:complexityL[x]} (and the usual complexity of the Hensel
  algorithm \cite[Theorem 15.11]{GaGe13}).
\end{proof}

\begin{cor}\label{cor:hensel-w}
  Let $F\in\Ai[x]$, $\tb$ a type dividing $F$, a factorisation
  $\Rt(F)=h_0\cdots\,h_r h_{\infty}$ and $n\in\Ni$. One can compute
  $F_0,\dots,F_r,F_\infty$ such that $\Rt(F_i)=h_i$ and
  $\wx(F-F_0\cdots{}F_r F_{\infty})>n+\wx(F)$ in
  $\Ot\left(\frac{n+\wx(F)}{\wx(\pi)}\,\dx\right)$
  operations in $\Fi$. If $\deg(F_\tb)>\dx/2$, this is
  $\Ot(n\,\dx/\wx(\pi)+\vF)$.
\end{cor}

\begin{proof}
  This algorithm is similar to the classical multifactor Hensel
  lifting \cite[Algorithm 15.17]{GaGe13}. We start by building a
  subproduct tree of the factorisation
  $\Rt(F)=h_0 \cdots\,h_r \,h_{\infty}$. Then, for each node (from top
  to bottom), we initialise the polynomials $G,H,U,V$, and run
  \onestep{} until we reach the required precision. Complexity and
  correctness follows from Propositions \ref{prop:init} and
  \ref{prop:hensel}. If $\deg(F_\tb)>\dx/2$, Lemma \ref{lem:trunc}
  implies $\frac{\wx(F)}{\wx(\pi)}<\frac{4\,\vF}{\dx}$.
\end{proof}

In order to get a factorisation
$F=F_0\cdots{}F_r F_{\infty} \mod \pi^{\sigma+1}$ for some given
$\sigma\in\Ni$, we need to relate the valuations $\val_0$ and
$\wx$:

\begin{lem}\label{lem:w-v0}
  Let $G\in\Li[x]$ and denote
  $N= \lfloor \deg(G)/\deg(\phi)\rfloor+1$. Then we have
  $\val_0(G)\wx(\pi)\le \wx(G)\leq \val_0(G)\wx(\pi) + N\wx(\phi)$.
\end{lem}

\begin{proof}
  First inequality is clear. For the second inequality, we can safely
  suppose that $\val_0(G)=0$. We first consider the case
  $\deg(G)<\deg(\phi)$ (so $N=1$) and we proceed by induction (remind
  that $\phi=\phi_k$ and $\wx=\val_{k+1}$). If $k=0$, the claim is
  obvious. Assume $k\ge 1$. Then, from \eqref{eq:defvk} and the
  assumption $\deg(G)<\deg(\phi_{k})$, we have
  $\val_{k+1}(G)=q_{k}\val_{k}(G)$. We also get
  $\val_{k+1}(\phi_{k})=q_{k}\val_{k}(\phi_{k})+m_{k}$. It is thus
  sufficient to show that $\val_{k}(G)\leq\val_{k}(\phi_{k})$. Write
  $G=\sum_{0\leq i<q_{k-1}\ell_{k-1}}a_i'\phi_{k-1}^i$.  As there is
  at least one $a_i'$ not dividable by $\pi$, i.e.
  $\val_{k}(a_i')\leq \val_{k}(\phi_{k-1})$ by recursion, we get
  $\val_{k}(G)\leq q_{k-1}\ell_{k-1}\val_{k}(\phi_{k-1})$. As
  $\val_{k}(\phi_{k})=q_{k-1}\ell_{k-1} \val_{k}(\phi_{k-1})$ by
  \cite[Theorem 2.11]{GuMoNa12}, this concludes. If
  $\deg(G)\ge \deg(\phi)$, we write $G=\sum_{i=0}^{N-1} a_i\,\phi^i$,
  $\deg(a_i)<\deg(\phi)$. Let $i$ such that $\vx_0(a_i)=0$. We get
  $\wx(G)\le \wx(a_i\phi^i) \le (i+1)\wx(\phi)\le N\wx(\phi)$.
\end{proof}

\begin{thm}\label{thm:slopefacto}
  There exists an algorithm \slopefacto{} that given $F\in\Ai[x]$,
  $\tb$ a type dividing $F$, a factorisation
  $\Rt(F)=h_0\cdots\,h_r h_{\infty}$ and $\sigma\in\Ni$, computes a
  factorisation $F=F_0\cdots{}F_r F_{\infty} \mod \pi^{\sigma+1}$ with
  $\Rt(F_i)=h_i$. It takes
  $\Ot\left(\sigma\,\dx + \frac{\dx^2\,\wx(\phi)}{\deg(\phi)\wx(\pi)}
  \right)$ operations in $\Fi$. If $\deg(F_\tb)>\dx/2$, this is
  $\Ot(\dx\,\sigma+\vF)$.
\end{thm}
\begin{proof}
  Let $n = \wx(\pi^\sigma)-\wx(F)+ (\lfloor \dx/\deg(\phi) \rfloor+1)
  \, \wx(\phi)$ and apply Corollary \ref{cor:hensel-w} with $n$. Then
  $G=F-F_0\cdots{}F_r F_{\infty}$ satisfies $\deg(G)<\dx$ and
  $\wx(G)>n+\wx(F)$, i.e. $\val_0(G) > \sigma$ by
  Lemma \ref{lem:w-v0}.  The first complexity bound follows from
  Corollary \ref{cor:hensel-w}. As $F_{\tb}$ is strongly monic in
  $\phi$, we have $\wx(\phi)/\deg(\phi)=\wx(F_{\tb})/\deg(F_{\tb})$
  and the second bound is a consequence of Lemma \ref{lem:trunc},
  which implies $\frac{\wx(F_\tb)}{\wx(\pi)}<\frac{2\,\vF}{\deg(F_\tb)}$.
\end{proof}

% \begin{rmk}
%   \adrien{Dire ici que cet algorithme a une meilleure complexit\'e que
%     le SFL de Nart et al.}
% \end{rmk}

%%% Local Variables:
%%% mode: latex
%%% TeX-master: "issac"
%%% End:

\section{A fast factorisation algorithm}
\label{sec:divideandconquer}
From the previous two sections, we can easily deduce a factorisation
algorithm. Let $n\geq\vF$ (\cite[Theorem 3.13]{BaNaSt13} shows that
this is a sufficient precision to detect the whole factorisation).
\begin{enumerate}
\item \label{firstfacto:irred} Run algorithm \fastirred{} with
  precision $2\vF/\dx$. We conclude that $F$ is irreducible (and
  return $F)$ or we get a type $\tb_{k-1}$.
\item Compute the factorisation $\Rt_k(F)=h_0 h_1\cdots h_r$ in
  $\Fi_k[y]$ (we have $h_{\infty}=1$ since $F$ is of type $\tb_{k-1}$).
\item Use Algorithm \slopefacto{} of Section \ref{sec:hensel} to get a
  factorisation $F=F_0 F_1\cdots F_r$ with precision $n$.
\item Go back to Step \ref{firstfacto:irred} for each $H_i$.
\end{enumerate}
If $\rho$ is the number of irreducible factors of $F$, this requires
at most $\rho$ univariate factorisations and $\log_2(\dx)$ univariate
irreducible tests over $\Fi_k[y]$, plus a number of operations over
$\Fi$ bounded by $\Oe(\rho\,n\,\dx)$ under Assumption
\ref{hyp:char>d}, and $\Oe(\rho\,n\,\dx+\vF^2)$ otherwise.

This section describes a divide and conquer strategy that will reduce
the computations over $\Fi$ to respectively $\Oe(\dx\,n)$ and
$\Oe(\dx\,n+\vF^2)$. The idea is the following:
\begin{enumerate}
\item Find a type $\tb$ such that $F_\tb$ has degree $>\dx/2$ and
  that either $F_\tb$ is irreducible, either any of its proper factor
  has degree $\leq\dx/2$.
\item Use Algorithm \slopefacto{} of Section \ref{sec:hensel} to get a
  factorisation $F=F_0\,F_1\cdots{}F_r F_{\infty}$ with
  $\F_\tb=F_0\,F_1\cdots{}F_r$.
\item Apply recursively this strategy to all factors that are not
  known irreducible.
\end{enumerate}

\subsection{Finding a dividing type}
\label{ssec:partial-facto}

The aim of this section is to either find the irreducible factor of
degree $>\dx/2$ if it exists, either find a factorisation
$F=F_0\,\F_1\cdots F_r\,F_\infty$ with each factor of degree
$\leq \dx/2$, and to do so with precision $\sigma\leq4\vF/\dx$. We
start with the algorithm.

\begin{algorithm}[ht]
  \nonl\TitleOfAlgo{\divtype($F, \sigma$)\label{algo:divtype}}
  \KwIn{%
    $F\in \Ai[x]$ monic separable, $\sigma\in\Ni$ the precision used.}%
  \KwOut{A type $\tb$ with the properties described above.}%
  $H\gets{} F$, $\dx\gets\deg(F)$\;%
  \While{True}{%
    $(b,\tb)\gets$ \fastirred($H,\sigma$)\;%
    \lIf{$b=$True}{\Return{$\tb$}\label{divtype:isirred}}%
    Compute $\Rt(H)=g\,h$, $h=h_i$ with highest degree\label{divtype:RtildeH}\;%
    % \tcp*[f]{$\Rt(H)=g\,h$, $h=h_i$ with highest degree}\;%
    \lIf{$\deg(h)\,\deg(\phi)\leq\dx/2$}{\Return{$\tb$}\label{divtype:testdegree}}%
    $(G,H)\gets$\slopefacto($H,\tb,[g,h],\sigma$)\label{divtype:slopefacto}\;%
  }%
\end{algorithm}
\begin{rmk}
  Note that the notation \fastirred($H,\sigma$) means to run Algorithm
  \fastirred{} with parameter $H$ and precision $\sigma$. Also, at
  Line \ref{divtype:RtildeH}, $H$ is always of type $\tb$, so that the
  factorisation $\Rt(H)$ does not involve $h_\infty$.
\end{rmk}

\begin{prop}\label{prop:divtype}
  The function call \divtype($F,4\vF/\dx$) is correct. It takes less
  than $\Oe(\rho\,\vF)$ operations in $\Fi$ if Assumption
  \ref{hyp:char>d} is satisfied, and $\Oe(\rho\,\vF+\vF^2)$ otherwise,
  plus at most $\rho$ factorisations and $\log_2(\dx)$ irreducibility
  tests over some $\Fi_k$.
\end{prop}

\begin{proof}
  Note first that $\deg(H)$ decreases at each loop (via the
  factorisation of Line \ref{divtype:slopefacto}), proving the ending
  of the algorithm. Also, the test of Line \ref{divtype:testdegree}
  ensures $\deg(H)>\dx/2$, so that $2\vF(H)/\deg(H)\le 4\vF/\dx$ at
  any point of the algorithm. \fastirred($H, 4\vF/\dx$) will thus
  return a correct answer thanks to Lemma \ref{lem:vk}. This precision
  is also sufficient to compute $\tilde{R}(H)$ at Line
  \ref{divtype:RtildeH} (use Proposition \ref{prop:Rk}). As the
  precision used is high enough, correctness is straightforward: we
  output $\tb$ when either $\Fi_\tb$ is irreducible (Line
  \ref{divtype:isirred}), either all of its factor has degree
  $\leq\dx/2$ (Line \ref{divtype:testdegree}). Finally, the complexity
  is a direct consequence of Theorems \ref{thm:fastirred} and
  \ref{thm:slopefacto} (and Proposition \ref{prop:GenIrr}).
\end{proof}

\begin{rmk}
  There are various straightforward improvements to this algorithm:
  for instance, one can provide the type computed so far to
  \fastirred{} to avoid some already done computations. They do not
  appear to make the reading easier, but are needed to run less
  than $\log_2(\dx)$ irreducible tests over some $\Fi_k[y]$.
\end{rmk}

\subsection{The divide and conquer algorithm}
\label{ssec:divide}

Thanks to this result, we derive a fast factorisation algorithm:

\begin{algorithm}[ht]
  \nonl\TitleOfAlgo{\fastfacto($F, n$)\label{algo:fastfacto}}
  \KwIn{%
    $F\in \Ai[x]$ monic separable, $n\in\Ni$ the precision.}%
  \KwOut{The list $F_1,\cdots,F_\rho$ of irreducible factors of $\Fi$ known with precision $n$.}%
  $\tb\gets$\divtype($F, 4\,n/d$)\;%
  Compute $\Rt(F)=h_0\,h_1\cdots h_r\,h_{\infty}$\;%
  $F_0,\cdots,F_r,F_\infty\gets$\slopefacto($F,\tb,[h_0,\cdots,h_r,h_{\infty}],n$)\label{fastfacto:slopefacto}\;%
  $\Rc\gets\{\}$\;%
  \For{$i\in\{0,\cdots,r,\infty\}$}{%
    \lIf{$F_i$ is known irreducible}{$\Rc\gets\Rc\cup\{F_i\}$}\label{fastfacto:knownirr}%
    \lElse{$\Rc\gets\Rc\cup\{$\fastfacto$(F_i,n)\}$}%
  }%
  \Return$\Rc$\;%
\end{algorithm}

\begin{rmk}
  At Line \ref{fastfacto:knownirr}, $F_0$ and $F_\infty$ are
  \emph{known irreducible} if they have degree $\le 1$ while $F_i$ is
  \emph{known irreducible} if $\ordre[\tb]{F_i}=1$ for $i=1,\ldots,r$.
\end{rmk}

\begin{thm}
  \label{thm:fastfacto}
  Let $n\geq\vF$. A function call \fastfacto($F, n$) returns the
  correct output. It requires at most $\rho$ factorisations and
  $\log_2(\dx)$ irreducibility tests in some $\Fi_k[y]$, plus
  $\Oe(\dx\,n)$ operations in $\Fi$ if Assumption \ref{hyp:char>d} is
  satisfied, and $\Oe(\dx\,n+\vF^2)$ otherwise.
\end{thm}

\begin{proof}
  This is a direct consequence of Proposition \ref{prop:divtype} and
  Theorem \ref{thm:slopefacto}, as $\sum_i\deg(F_i)=\dx$ at Line
  \ref{fastfacto:slopefacto}, and $\deg(F_i)<\dx/2$, so that the
  number of lines of the recursive calls tree is less than
  $\log_2(\dx)$. Precision $\vF$ is sufficient to detect all
  irreducible factors from \cite[Theorem 3.13]{BaNaSt13}.
\end{proof}

\begin{rmk}
  As in Section \ref{sec:irred}, we do not know in advance $\vF$. We
  proceed similarly as explained in Section \ref{ssec:prec}.
\end{rmk}

When considering $n\in \O(\vF)$, then up to the cost of the residual
factorisations, Theorem \ref{thm:fastfacto} improves \cite[Theorem
5.15]{BaNaSt13} by a factor $\vF$ if Assumption \ref{hyp:char>d} is
satisfied, and by a factor $\min(\dx,\vF)$ otherwise.

\subsection{Residual factorisations over finite fields} If
$\Card(\Fi)=q$, we factorise $F$ mod $\pi$ with an expected
$\Ot(\dx^2+\dx\,\log(q))$ operations over $\Fi$ by \cite[Corollary
14.30]{GaGe13}. The remaining residual factorisations performed by
algorithm \fastfacto{} use then an expected $\Oe(\dx\vF\log(q))$
operations in $\Fi$, see the proof of \cite[Theorem 5.14]{BaNaSt13}.
Together with Theorem \ref{thm:fastfacto}, this proves:% the following
%result:

\begin{cor}\label{cor:factoZp}
  Let $F\in \Zi_p[x]$, separable with $p$ small. Given the univariate
  factorisation of $F\mod p$, we compute the irreducible factors of
  $F$ with precision $n\ge \vF$ with an expected $\Oe(n\dx)$
  operations over $\Fi_p$ if $p> \dx$ and $\Oe(n\dx+\vF^2)$
  otherwise. The same result holds for $F\in \Fi_p[[t]][x]$.
\end{cor}

Corollary \ref{cor:factoZp} improves significantly \cite[Theorem
5.18]{BaNaSt13} which leads to the complexity estimate
$\Oe(\dx\vF^2+n \dx^2)$ under the same hypothesis.

\begin{rmk}\label{rmk:residualfacto} Corollary
  \ref{cor:factoZp} requires \textit{a priori} to compute a primitive
  representation of $\Fi_k$ over $\Fi$ before applying the
  factorisation algorithm \cite[Corollary 14.30]{GaGe13} (this issue
  doesn't seem to be considered in \cite{BaNaSt13}). To this aim, we
  may use a Las Vegas subroutine \cite[Proposition 4]{PoWe17} whose
  complexity fits in our aimed bound. There are recent faster
  factorisation algorithms \cite[Corollary 3]{HoLe20} (both for
  primitive or triangular representation), but not yet implemented.
\end{rmk}

\subsection{Avoiding residual factorisations}
\label{ssec:residualFacto}

If $\Fi$ has large cardinal, then the residual factorisations will
probably dominate the cost of the all algorithm. It might thus be
preferable to rely on dynamic evaluation \cite{HoLe19,PoWe17} : we
allow the $P_k$'s to be square-free, hence the $\Fi_k$'s to be product
of fields. If at some point we find a zero divisor (while computing
some gcd's), then we pursue over each discovered summand of $\Fi_{k}$
(or we return false if we run an irreducibility test). At the end, we
perform a unique residual factorisation (of expected small degree) for
each discovered factor of $F$ to deduce its irreducible
factorisation. Notice that this last step might be useless, depending
on the arithmetic information we want to compute. It is not needed for
instance if we only want the ramification indices (e.g. for the
computation of the genus of a plane curve \cite{PoWe17}), or the
valuation $\vF$ of the discriminant of $F$ \cite{Na14}, or the
equisingularity type of a germ of plane curve \cite{PoWe20}.  In these
cases, we may only use square-free residual factorisations (fast
gcd's) and the underlying algorithm is entirely deterministic.

%%% Local Variables:
%%% mode: latex
%%% TeX-master: "issac"
%%% End:

\section{Direct Applications}
\label{sec:applications}
Our complexity estimates have several direct consequences for
various tasks of computational number theory and algebraic
geometry. In what follows, the complexity results are given up to the
cost of the residual univariate factorisations.

\paragraph{OM-factorisation} Let $\Ki=\Qi[x]/(F)$ be a number field
and let $p\in \Zi$ a prime. The first main consequence of Theorem
\ref{thm:fastfacto} is that we compute an \textit{Okutsu-Montes (OM)
  representation} of the prime ideals dividing $p$ with $\Oe(\dx \vF)$
operations in $\Fi_p$ if $p>\dx$ or $\Oe(\dx \vF+\vF^2)$ otherwise,
improving \cite[Theorem 5.15]{BaNaSt13} respectively by a factor $\vF$
or $\min(\dx,\vF)$. These OM-representations carry on essential data
about the corresponding extensions of local fields and give useful
tools for various local and global arithmetic tasks (see
e.g. \cite{GuMoNa11}).

\paragraph{Valuations of discriminants and resultants} As a
straightforward application of fast OM-factorisation, we may use
\cite{Na14} to compute $\vF=\val_p(\disc[]{F})$ with $\Oe(\dx\vF)$
operations in $\Fi_p$ if $p>\dx$ or $\Oe(\dx \vF+\vF^2)$ otherwise,
improving \cite[Theorem 2.5]{Na14} respectively by a factor $\vF$ or
$\min(\dx,\vF)$. If $G,H\in \Zi[x]$ are coprime of degrees at most
$\dx$, we compute $\delta=\val_p(\res[](G,H))$ within the same bound,
improving now \cite[Theorem 3.3]{Na14} by a factor $\delta$ or
$\min(\dx,\vF)$. As mentioned in Section \ref{ssec:residualFacto}, we
may rely on dynamic evaluation for this task: only square-free
residual factorisation is required and the algorithm is deterministic.

\paragraph{Local integral basis} Combined with Bauch's algorithm
\cite{Ba16}, our results allow to compute a $\Zi_{(p)}$-basis of the
integral closure of $\Zi_{(p)}$ in $\Ki$ with $\Oe(\dx^2\vF)$
operations in $\Fi_p$ if $p>\dx$ or $\Oe(\dx^2 \vF+\vF^2)$ otherwise,
improving $\Oe(\dx^{2}\vF+\dx\vF^{2})$ of \cite[Lemma
3.10]{Ba16}. This impacts the computation of a global integral basis
of $\Ki/\Qi$, obtained from local ones via Hermite normal forms and
Chinese remaindering.

\paragraph{Function fields} All complexity results above extend
trivially to function fields satisfying Assumption
\ref{hyp:char>d}. In this context, we may use moreover the
Riemann-Hurwitz formula to compute the genus of a degree $\dx$ plane
curve over $\Fi$ within $\Oe(\dx^3)$ operations in $\Fi$ (following a
similar strategy than in \cite{PoWe17}, but using the easier to
implement algorithm \fastfacto{}). Here again, only square-free
residual factorisation is required and the algorithm is deterministic.

%\paragraph{Computing roots of polynomials} It's straightforward to
%adapt our algorithm in order to compute only the degree one
%polynomials in $\Ai[x]$ dividing $F$ mod $\pi^n$, getting in such a
%way all $\theta\in \Ai$ with precision $n$ such that
%$\val(F(\theta))\ge n$. If $\Ai=\Zi_p$, we use $\Ot(\dx n)$ operations
%in $\Fi_p$, extending results of \cite{NeRoSc17} which considers rings
%of power series. Notice that in this context, no residue field
%extension appears and we compute only degree one factors over $\Fi$,
%so the algorithm is purely deterministic. Moreover, as in
%\cite{NeRoSc17}, there is no need to assume $F$ separable and multiple
%roots are allowed.

\paragraph{More applications} There are several other computational
consequences for global fields, as computing the valuation at a prime
ideal, factoring fractional ideals or Chinese remaindering
\cite{Na11}, factoring bivariate polynomials \cite{We16}, computing
roots of polynomials \cite{NeRoSc17}, or computing Riemann-Roch spaces
\cite{He02}, but going into these details is beyond the scope of this
paper.

%%% Local Variables:
%%% mode: latex
%%% TeX-master: "issac"
%%% End:

\bibliographystyle{abbrv} {\bibliography{tout}}
\addcontentsline{toc}{section}{References.}

\end{document}